\documentclass[conference,10pt,letterpaper]{IEEEtran}

\usepackage{amsmath,amssymb,amsthm}
\usepackage{xifthen}
\usepackage{mathtools}
\usepackage{enumerate}
\usepackage{microtype}
\usepackage{xspace}
\usepackage{bm}
\usepackage[T1]{fontenc}
\usepackage{fancyhdr}
\usepackage{lastpage}
\usepackage{bbm}

\newcommand{\cond}{\,\vert\,}
\renewcommand{\Re}[1][]{\ifthenelse{\isempty{#1}}{\operatorname{Re}}{\operatorname{Re}\left(#1\right)}}
\renewcommand{\Im}[1][]{\ifthenelse{\isempty{#1}}{\operatorname{Im}}{\operatorname{Im}\left(#1\right)}}

\newcommand{\rv}{\vect{r}}

\newcommand{\Ac}{{\mathcal A}}

\newcommand{\Sc}{{\mathcal S}}

\newcommand{\NN}{\mathbb{N}}

\newcommand{\CN}[1][]{\ifthenelse{\isempty{#1}}{\mathcal{N}_{\mathbb{C}}}{\mathcal{N}_{\mathbb{C}}\left(#1\right)}}

\renewcommand{\P}[1][]{\ifthenelse{\isempty{#1}}{\mathbb{P}}{\mathbb{P}\left(#1\right)}}
\newcommand{\E}[1][]{\ifthenelse{\isempty{#1}}{\mathbb{E}}{\mathbb{E}\left[#1\right]}}
\newcommand{\I}[1][]{\ifthenelse{\isempty{#1}}{\mathbb{I}}{\mathbb{I}\left\{#1\right\}}}
\renewcommand{\det}[1][]{\ifthenelse{\isempty{#1}}{\mathrm{det}}{\mathrm{det}\left(#1\right)}}
\newcommand{\trace}[1][]{\ifthenelse{\isempty{#1}}{{\rm tr}}{\mathrm{tr}\left(#1\right)}}
\newcommand{\rank}[1][]{\ifthenelse{\isempty{#1}}{\mathrm{rank}}{\mathrm{rank}\left(#1\right)}}
\newcommand{\diag}[1][]{\ifthenelse{\isempty{#1}}{\mathrm{diag}}{\mathrm{diag}\left(#1\right)}}
\newcommand{\Cov}[1][]{\ifthenelse{\isempty{#1}}{\mathsf{Cov}}{\mathsf{Cov}\left(#1\right)}}

\newcommand{\defeq}{\triangleq}

\newtheorem{proposition}{Proposition}
\newtheorem{example}{Example}

\newcounter{enumi_saved}
\usepackage{answers}
\Newassociation{solution}{Solution}{solutionfile}

\AtBeginDocument{\Opensolutionfile{solutionfile}[\jobname]}
\AtEndDocument{\Closesolutionfile{solutionfile}\clearpage
}

\IfFileExists{MinionPro.sty}{
}{
}

\usepackage{subfigure}
\usepackage{graphicx,psfrag}
\usepackage{multirow}
\usepackage{array}
\usepackage{setspace}
\usepackage{amsmath,bbm,mathtools}
\usepackage{color}
\usepackage{enumitem}
\usepackage{stfloats}
\usepackage[flushleft]{threeparttable}
\usepackage{pgfplots}
\usepackage{booktabs}
\pgfplotsset{minor grid style={dotted}}
\pgfplotsset{major grid style={dashed}}
\usepackage{tikz}
\usetikzlibrary{fit,positioning,arrows,shapes,shapes.multipart,calc,arrows.meta,shapes.geometric}
\usepackage[ruled,vlined,linesnumbered]{algorithm2e}
\usepackage[toc,acronym]{glossaries}
\makeindex

\mathtoolsset{showonlyrefs=true}

\usepackage{grffile}
\pgfplotsset{compat=newest}
\usetikzlibrary{plotmarks}
\usetikzlibrary{arrows.meta}
\usepgfplotslibrary{patchplots}

\renewcommand{\rv}[1]{{\mathsf{#1}}}

\newcommand{\of}[1]{^{(#1)}}

\newcommand*\dif{\mathop{}\mathrm{d}}

\renewcommand{\defeq}{=}

\DeclareMathOperator*{\minimize}{minimize}

\newcolumntype{C}[1]{>{\centering\let\newline\\\arraybackslash}m{#1}}

\makeatletter
\newcommand{\removelatexerror}{\let\@latex@error\@gobble}
\makeatother

\title{Age of Information in Prioritized Random Access \vspace{-.3cm}} 
  
\author{
	\IEEEauthorblockN{Khac-Hoang Ngo, Giuseppe Durisi, and Alexandre Graell i Amat} 
	\IEEEauthorblockA{Department of Electrical Engineering, Chalmers University of Technology, 41296 Gothenburg, Sweden}
		\thanks{This project has received funding from the European Union's Horizon 2020 research and innovation programme under the Marie Skłodowska-Curie grant agreement No 101022113.}
		\vspace{-1cm}
}

\newacronym{MAC}{MAC}{multiple access channel}
\newacronym{UMRA}{UMRA}{unsourced massive random access}
\newacronym{SIMO}{SIMO}{single-input multiple-output}
\newacronym{SISO}{SIMO}{single-input single-output}
\newacronym{iid}{i.i.d.}{independent and identically distributed}
\newacronym{ML}{ML}{maximum likelihood}
\newacronym{PEP}{PEP}{pair-wise error probability}
\newacronym{LLR}{LLR}{log-likelihood ratio}
\newacronym{SNR}{SNR}{signal-to-noise ratio}

\newacronym{AoI}{AoI}{age of information}
\newacronym{AVP}{AVP}{age-violation probability}
\newacronym{PMF}{PMF}{probability mass function}
\newacronym{SA}{SA}{slotted ALOHA}
\newacronym{IRSA}{IRSA}{irregular repetition slotted ALOHA}
\newacronym{SIC}{SIC}{successive interference cancellation}
\newacronym{PLR}{PLR}{packet loss rate}
\newacronym{DE}{DE}{density evolution}
\newacronym{IoT}{IoT}{Internet of Things}

\IEEEoverridecommandlockouts
\begin{document}

\maketitle
\begin{abstract} 
	Age of information (AoI) is a performance metric that captures the freshness of status updates.
	While AoI has been studied thoroughly for point-to-point links, the impact of modern random-access protocols on this metric is still unclear.
	In this paper, we extend the recent results by Munari to prioritized random access where devices are divided into different classes according to different AoI requirements.
	We consider the irregular repetition slotted ALOHA protocol and analyze the AoI evolution by means of a Markovian analysis following similar lines as in Munari (2021). 
	We aim to design the protocol to satisfy the AoI requirements for each class while minimizing the power consumption. To this end, we optimize the update probability and the degree distributions of each class, such that the probability that their AoI exceeds a given threshold lies below a given target and the average number of transmitted packets is minimized. 
\end{abstract}
%

\vspace{-.05cm}
\section{Introduction} \label{sec:intro}
The \gls{IoT} foresees a very large number of devices, which we will refer to as users, to be connected and exchange data in a sporadic and uncoordinated manner. This has led to the development of modern random access protocols~\cite{Berioli2016NOW}. In most of these protocols, the users transmit multiple copies of their packets to create time diversity, and the receiver employs \gls{SIC} to decode.  In particular, in the \gls{IRSA} protocol~\cite{Liva2011}, the users draw the number of copies from a degree distribution and transmit the copies in randomly chosen slots of a fixed-length frame. A common design goal is to minimize the \gls{PLR}, thus maximizing the chance to deliver packets to the receiver successfully.

In many IoT applications, it is becoming increasingly important to deliver packets successfully and to guarantee the timeliness of those packets simultaneously. Examples include sensor networks, vehicular tracking, and health monitoring. In these delay-sensitive applications, the packets carry critical status updates that are required to be fresh. The \gls{AoI} metric  (see, e.g.,~\cite{Kosta2017} and references therein)  has been introduced precisely to account for the freshness of the status updates. It captures the offset between the generation of a packet and its observation time. In~\cite{Yates2019}, the \gls{AoI} in a system where independent devices send status updates through a shared queue was analyzed. The \gls{AoI} has been used as a performance metric to design status update protocols in, e.g.,~\cite{Jiang2019,Gu2019}. The first analytical characterization of the \gls{AoI} for a class of modern random access, namely \gls{IRSA}, has been recently reported in~\cite{Munari2020modern}.
 
Since \gls{IoT} devices are mostly battery-limited, their power consumption should be minimized. By assuming that each packet transmission consumes a fixed amount of energy, we can use the average number of transmitted packets  per slot as a proxy of the power consumption. When status updates are conveyed via an \gls{IRSA} protocol, the number of transmitted packets per user depends both on the user activity, i.e., on how often the user has an update to transmit, and on the degree distribution assigned to the user. This leads to a tension between minimizing the AoI and minimizing the number of packets. Too sporadic user activity leads to stale updates, but too frequent updates 
lead to channel congestion and update failure. Furthermore, degree distributions with low degrees lead to a low number of transmitted packets but high PLR, which results in larger \gls{AoI}, while degree distributions with high degrees achieve low PLR at the cost of a larger number of transmitted packets. Therefore, the user update probability and the degree distribution need to be carefully selected.

In this paper, we consider an IoT monitoring system where users attempt to deliver timely status updates to a receiver following the \gls{IRSA} protocol. We assume that the users are heterogeneous and their updates require different levels of freshness. Accordingly, users are divided into different classes, each with a different \gls{AoI} requirement. Following similar lines as in~\cite{Munari2020modern}, we analyze the \gls{AoI} evolution by means of a Markovian analysis and derive the \gls{AVP}, i.e., the probability that the AoI exceeds a certain threshold, for each class. We study the trade-off between the \gls{AVP} and the number of transmitted packets by investigating the impact of the update probability and the degree distributions. Since the \gls{PLR} of \gls{IRSA} and, hence, the \gls{AVP} are not known in closed form, we propose an easy-to-compute PLR approximation, which leads to an accurate approximation of the \gls{AVP}. Our PLR approximation is based on \gls{DE}~\cite{Liva2011} and on existing \gls{PLR} approximations in the error-floor region~\cite{Ivanov2017} and the waterfall region~\cite{GraelliAmat2018}. We jointly optimize the update probability and the degree distributions for each class to minimize the number of transmitted packets while guaranteeing that the \gls{AVP} of each class lies below a given target. Our simulation results show that the number of transmitted packets can be significant reduced with optimized irregular degree distributions, compared to regular distributions. Our experiments also suggest that using degrees up to $3$ is sufficient for a setting where there are two classes containing respectively $800$ and $3200$ users, the framelength is $100$ slots, and the AoI of class-$1$ users and class-$2$ users exceeds a threshold $7.5 \!\times\! 10^4$ and $4.5 \!\times\! 10^4$ with probability as low as $10^{-5}$ and $10^{-3}$, respectively. 

\section{System Model and Problem Formulation}
We consider a system with $U$ users attempting to deliver timestamped status updates to a receiver through a wireless channel. Time is slotted and each update is transmitted in a slot. We let the slot length be $1$ without loss of generality. Each user belongs to one of $K$ classes with different \gls{AVP} requirements. 
 Let $U_k$ be the number of users in class $k$, $k\in [K]$.\footnote{We use $[m:n]$ to denote the set of integers from $m$ to $n$, and $[n] \defeq [1:n]$.} We define the fraction of class-$k$ users as $\gamma_k = U_k/U$. We assume that a class-$k$ user has a new update in each slot with probability $\mu_k$ independently of the other users. 
We further assume that slots containing a single packet always lead to successful decoding, whereas slots containing multiple packets (or, more specifically, unresolved collisions after~\gls{SIC}) lead to decoding failures.

\subsection{Irregular Repetition Slotted ALOHA}

We assume that the system operates according to the \gls{IRSA} protocol. Time is divided into frames of $M$ slots and users are frame- and slot-synchronous. A user may generate more than one update during a frame, but only the latest update is transmitted in the next frame. An active user in class $k$ sends~$\rv{L}_k$ identical replicas of its latest update in $\rv{L}_k$ slots chosen uniformly without replacement from the $M$ available slots. The number $\rv{L}_k$ is called the degree of the transmitted packet. It follows a class-dependent probability distribution $\{\Lambda\of{k}_\ell\}$ where $\Lambda\of{k}_\ell \!\defeq \Pr[\rv{L}_k \!=\! \ell]$. We  write this distribution using a polynomial notation as
$\Lambda\of{k}(x) = \sum_{\ell=0}^d\Lambda\of{k}_\ell x^\ell
$
where $d$ is the maximum degree. Note that $\{\Lambda\of{k}\}$ may contain degree $0$. When $\rv{L}_k = 0$, the user discards the update.
Upon successfully receiving an update, the receiver is assumed to be able to determine the position of its replicas. In practice, this can be done by including in the header of the packet containing each update a pointer to the position of its replicas. The receiver employs a SIC decoder. It seeks slots containing a single packet, decodes the packet, then locates and removes the replicas. These steps are repeated until no slots with a single packet can be found. 

Note that a user in class $k$ has a new update in a frame with probability $\sigma_k = 1 - (1\!-\!\mu_k)^M$. Therefore, the number of class-$k$ users transmitting over a frame is a binomial random variable of parameters $(U_k,\sigma_k)$ with expected value $U_k \sigma_k$. The average channel load of class $k$ is given by
$
	G_k = U_k \sigma_k / M. 
$
The overall average channel load is $G = \sum_{k=1}^{K} G_k$. 

The average number of packets transmitted by a class-$k$ user per slot is
$
	\Phi_k = \sigma_k \dot{\Lambda}\of{k}(1)/M,
$
where $\dot{\Lambda}\of{k}(x)$ denotes the first-order derivative of $\Lambda\of{k}(x)$. The  total average number of transmitted packets per slot is given by
	\vspace{-.1cm} 
\begin{align}
	\Phi = \sum_{k=1}^K U_k \Phi_k = \sum_{k=1}^K G_k \dot{\Lambda}\of{k}(1).
\end{align}
\vspace{-.05cm} 
We 
use $\Phi$ as a proxy of the total power consumption. 

\subsection{Age of Information}
We define the \gls{AoI} for user $i$ at slot $n$ as
	$\delta_i(n) \defeq n - t_i(n)$,
where $t_i(n)$ denotes the timestamp of the last received update from user $i$ as of slot $n$. Since the \glspl{AoI} of users in the same class are stochastically equivalent, we denote a representative of the \glspl{AoI} of  class-$k$ users as $\delta\of{k}(n)$. The \gls{AoI} grows linearly with time and is reset at the end of a frame only when a new update is successfully decoded. 
We are interested in the value of the \gls{AoI} at the end of a generic frame $j \in \NN_0$. We will refer to this quantity simply as \gls{AoI} hereafter. 
For class $k$, this quantity is given by $\delta\of{k}(jM) + M$.
We define the \gls{AVP} as the probability that the \gls{AoI} exceeds a certain threshold $\theta$ at steady state. Specifically, the \gls{AVP} for class $k$ is defined as 
\vspace{-.05cm}
\begin{align} \label{eq:def_AVP}
	\zeta\of{k}(\theta) \defeq \lim\limits_{j\to\infty}\Pr[\delta\of{k}(jM) + M > \theta].
\end{align}
We shall see in the next section that the \gls{AoI} process is ergodic Markovian, thus the limit in~\eqref{eq:def_AVP} exists. 
We consider the requirement that the \gls{AoI} at steady state of class $k$ exceeds a threshold $\theta_k$ with probability no larger than $\epsilon_k$:
\vspace{-.05cm}
\begin{align} \label{eq:AoI_requirement}
	\zeta\of{k}(\theta_k) \le \epsilon_k, \quad k \in [K].
\end{align}

\subsection{Problem Formulation} \label{sec:formulation}
\vspace{-.1cm}
Our goal is to design the update probabilities $\{\mu_k\}$ and the degree distributions $\{\Lambda\of{k}\}$ such that the AoI requirements in~\eqref{eq:AoI_requirement} are satisfied, while the number of packets $\Phi$ is minimized:
\begin{align} \label{eq:power_minimization}
	\minimize_{\{\mu_k, \Lambda\of{k}(x)\}_{k=1}^K}~  &\Phi \qquad 
	\text{~subject to~} \eqref{eq:AoI_requirement}.
\end{align}


\vspace{-.15cm}
\section{\gls{AoI} Analysis}
\vspace{-.15cm}

\subsection{Current \gls{AoI}}
\vspace{-.05cm}
Let $P\of{k}$ denote the PLR of a class-$k$ user. It is convenient to denote  by $\xi_k = \sigma_k(1-P\of{k})$ the probability that the \gls{AoI} $\delta\of{k}(n)$ 
	is reset.
	Also, let $\rv{B}_k \in [M]$ denote the number of slots between the generation of a packet of a class-$k$ user and the start of the subsequent frame, when the user can access the channel. It has probability mass function
	$
		\Pr[\rv{B}_k = b] = \mu_k(1-\mu_k)^{b-1}/\sigma_k,
	$
	where the numerator is the probability that the user has generated an update for the last time $b$ slots before the end of a frame, and the denominator is the probability that at least one update is generated during the frame. Whenever an update is successfully decoded, the current \gls{AoI} is reset to $\rv{B}_k+M \in [M+1:2M]$.

In what follows, it will be convenient to decompose an arbitrary integer $n$ as $n = \alpha_n + M \beta_n$,
where $\alpha_n \defeq n \! \mod M$ and $\beta_n \defeq \lfloor n/M \rfloor$. Using this decomposition, we can write
$
	\delta\of{k}(n) = \delta\of{k} (M \beta_n) + \alpha_n,
$
where the first term on the right-hand side captures the age at the beginning of the current frame, and the second is the offset from the start of the current frame up to the observation time~$n$. We set $n = 0$ right after the reception of the first update. Thus, the initial \gls{AoI} is in $[M+1:2M]$, and $\delta\of{k}(n) \ge M+1, \forall n$. Therefore, the evolution of the  \gls{AoI} of a class-$k$ user is fully characterized by the discrete-time, discrete-valued stochastic process
\vspace{-.05cm}
\begin{align} \label{eq:def_Omega}
	\Omega\of{k}_j \defeq \delta\of{k}(jM) - (M+1), \quad j\in \NN_0, k\in [K],
\end{align}
where $j$ is the frame index.
Since each user operates independently over successive frames, $\Omega\of{k}_{j}$ is a Markovian process across $j$. The one-step transition probabilities $q\of{k}_{n_1,n_2} \defeq \Pr\big[\Omega\of{k}_{j+1} \!=\! n_2 \cond \Omega\of{k}_{j} \!=\! n_1\big]$ are given by 
\vspace{-.05cm}
\begin{align} \label{eq:transitio_prob}
	q\of{k}_{n_1,n_2} = \begin{cases}
		\xi_k \Pr[\rv{B}_k = n_2+1], &\text{for~} n_2 \in [0:M-1], \\
		1- \xi_k, &\text{for~} n_2 = n_1+M, \\
		0, &\text{otherwise.}
	\end{cases}
\end{align}
To verify \eqref{eq:transitio_prob}, note that $\Omega\of{k}_j$ is reset with probability $\xi_k$, and in this case, it is reset to a value $n_2+1$, where $n_2\in [0:M-1]$, with probability $\Pr[\rv{B}_k = n_2+1]$. With probability $1-\xi_k$, the variable $\Omega\of{k}_j$ is simply incremented by the framelength $M$. 

We start with the following observation.
\vspace{-.1cm}
\begin{proposition} \label{prop:steady_state_prob_IRSA}
	The stochastic process $\Omega\of{k}_j$ is ergodic, and has steady-state distribution
	\begin{align}
		\pi\of{k}_w = \xi_k \left(1-\xi_k\right)^{\beta_w} \Pr[\rv{B}_k=\alpha_w+1], \quad w \in \NN_0. 
	\end{align}
\end{proposition}
\begin{proof}
	The proof follows directly from the proof of the single-class case in \cite[Prop.~1]{Munari2020modern}.
\end{proof}
It follows from Proposition~\ref{prop:steady_state_prob_IRSA} that the limit in~\eqref{eq:def_AVP} exists.


\subsection{Age-Violation Probability}
It follows from~\eqref{eq:def_Omega} and \eqref{eq:def_AVP} that 
$
	\zeta\of{k}(\theta) 
	= \Pr[\Omega\of{k} > \theta - 2M - 1],
$
where the random variable $\Omega\of{k}$ has steady-state distribution $\{\pi_w\of{k}\}$. The following result holds.

\vspace{-.1cm}
\begin{proposition} \label{prop:peak_age_violation}
	The \gls{AVP} 
	is given by 
	\begin{multline}
		\!\!\!\!\!\zeta\of{k}(\theta) = \notag \\ \begin{cases}
			(1-\xi_k)^{\beta_{\theta-2M}} \left[1-\frac{1-(1-\mu_k)^{1+\alpha_{\theta-2M}}}{\sigma_k}\xi_k\right], &\text{for~} \theta > 2M, \\
			1, &\text{otherwise}.
			\end{cases}\label{eq:AVP}
	\end{multline}
\end{proposition}
\begin{proof}
	The proof follows similar steps as the proof of the single-class case in~\cite[Prop.~3]{Munari2020modern}. 
\end{proof}

\begin{example} \label{example}
	Consider a system with $U= 4000$ users, framelength $M = 100$, $K = 2$ classes with fractions $(\gamma_1,\gamma_2) = (0.2, 0.8)$, AoI thresholds $(\theta_1,\theta_2) = (7.5 \times 10^4,4.5\times 10^4)$, and target \glspl{AVP} $(\epsilon_1,\epsilon_2) = (10^{-4},10^{-2})$. We further assume that $\mu_1 \!=\! \mu_2 \!=\! \mu$. Thus, $\Phi = \frac{U(1-(1-\mu)^M)}{M} \sum_{u = 1}^U \gamma_k \dot{\Lambda}\of{k}(1)$, which increases with $\mu$.  We evaluate the \glspl{AVP} for this scenario and plot them as functions of $\Phi$ in Fig.~\ref{fig:AVP_vs_power} for three sets of regular degree distributions, namely, $\Lambda\of{1}(x) = \Lambda\of{2}(x) \in \{x,x^2,x^3\}$. 
	The PLR is computed numerically. We vary $\Phi$ by varying $\mu$. Some remarks are in order.
	\begin{itemize}[leftmargin=*]
		\item For each class, the \gls{AVP} first decreases and then increases with $\Phi$. Indeed, when $\mu$ is low, collisions are unlikely. Although the updates are successfully received with high probability, the sporadicity of the updates entails a high AoI. When $\mu$ is high, users transmit frequently, and updates fail with high probability due to collision, entailing a high AoI.
		
		
		\item The target \glspl{AVP} $(\epsilon_1,\epsilon_2) = (10^{-4},10^{-2})$ are not met when $\Lambda\of{1}(x) = \Lambda\of{2}(x) = x$. The distributions $\Lambda\of{1}(x) \!=\! \Lambda\of{2}(x) \!=\! x^2$ satisfy these requirements with a minimum number of packets per slot $\Phi \approx 1.09$. The distributions $\Lambda\of{1}(x) \!=\! \Lambda\of{2}(x) \!=\! x^3$ require a higher $\Phi$ to achieve the same requirements, but can yield a reduction of the \gls{AVP}. For example, the more stringent requirements $(\epsilon_1,\epsilon_2) = (10^{-5},10^{-3})$ can be met with about $1.92$ packets/slot. In general, distributions with low degrees can achieve mild \gls{AoI} requirements with a low $\Phi$, while higher degrees are needed to achieve more stringent requirements. 
	\end{itemize}
\end{example}
The observations in Example~\ref{example} reveal the existence of a trade-off in the choice of $\{\mu_k\}$ and  $\{\Lambda\of{k}\}$ to satisfy the AoI requirements while minimizing $\Phi$. 

\begin{figure} [t!]
	\centering
	\vspace{-.15cm}
	\input{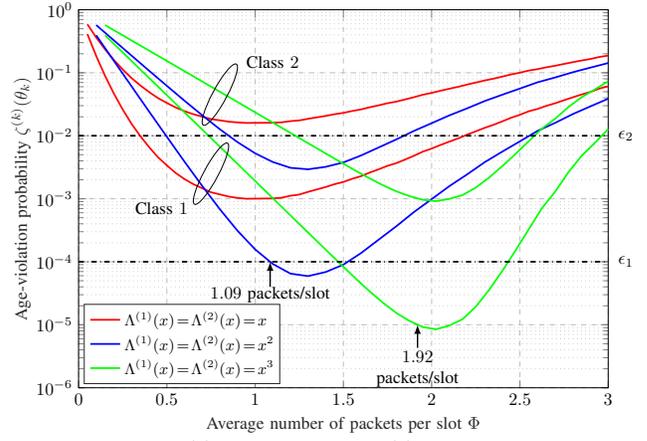}
	\vspace{-.35cm}
	\caption{The \glspl{AVP} $\zeta\of{1}(7.5 \times 10^4)$ and $\zeta\of{2}(4.5 \times 10^4)$ vs. $\Phi$ for the scenario in Example~\ref{example} with $(\epsilon_1,\epsilon_2) = (10^{-4},10^{-2})$. We consider three sets of regular degree distributions, namely, $\Lambda\of{1}(x) = \Lambda\of{2}(x) \in \{x,x^2,x^3\}$. 
	}
	\label{fig:AVP_vs_power}
	\vspace{-.6cm}
\end{figure}
\section{Packet Loss Rate Approximation} \label{sec:PLR_approx}
\vspace{-.1cm}
The PLR for class-$k$ users can be derived as~\cite[Eq.~(2)]{Ivanov2016}
\vspace{-.1cm}
\begin{equation}
	P\of{k} = \sum_{\ell= 0}^d \Lambda_\ell\of{k}P_\ell,
\end{equation}
where $P_\ell$ is the probability that a degree-$\ell$ user (of any class) is not resolved. The probability $P_\ell$ is determined by the overall channel load $G$ and the average degree distribution $\Lambda(x) = \sum_{\ell = 0}^d  \Lambda_\ell x^\ell$ with $\Lambda_\ell = \sum_{k=1}^{K} \gamma_k \Lambda\of{k}_\ell$, $\ell \in [0:d]$. If $\Lambda_0 > 0$, then $P_0 = 1$ and $P_\ell$, $\ell \ge 1$, is the probability that a degree-$\ell$ user is not resolved in a system with channel load $\bar{G} = G(1-\Lambda_0)$ and degree distribution $\bar{\Lambda}(x) \!=\! \sum_{\ell = 1}^d  \bar{\Lambda}_\ell x^\ell$ with $\bar{\Lambda}_\ell \!=\! \frac{1}{1-\Lambda_0} \sum_{k=1}^{K} \gamma_k \Lambda_\ell\of{k}$. Therefore, we assume without loss of generality that $\Lambda_0 = 0$ in the remainder of the section. The \gls{PLR} is not known in closed form in general, but can be computed numerically. However, since the optimization~\eqref{eq:power_minimization} requires repeated evaluation of the \gls{AVP}, and thus of the \gls{PLR}, simulation-based PLR computation becomes inefficient. Therefore, we seek an easy-to-compute approximation of the \gls{PLR} that leads to an accurate approximation of the \gls{AVP}. 

The SIC process of \gls{IRSA} is equivalent to graph-based iterative erasure decoding of low-density parity-check~(LDPC) codes. In the asymptotic regime where $M\to\infty$, $P_\ell$ can be evaluated using DE as~\cite{Liva2011}
\vspace{-.05cm}
\begin{equation}
	P_{\ell, \text{DE} } = \lim\limits_{i\to\infty} (\eta_i)^\ell,
\end{equation}
where $\eta_i$ is the probability that an edge connected to a degree-$\ell$ user remains unknown in the decoding process. It can be computed in an iterative manner as $\eta_0 = 1$, $\eta_i = 1- \exp(-G\dot{\Lambda}(\eta_{i-1}))$ where $\dot{\Lambda}(x) = \dif \Lambda(x)/\dif x$. As $M\to\infty$, $P_\ell = P_{\ell, \text{DE} }$ and the \gls{PLR} $P\of{k}$ drops at a certain threshold value as the channel load $G$ decreases. That is, all but a vanishing fraction of the class-$k$ users are resolved if the channel load is below the decoding threshold. According to~\cite[Prop.~1]{Ivanov2016}, the thresholds for all classes coincide and can be obtained by means of DE as the largest value $G^*$ of $g$ such that $\nu > 1-\exp(-g \dot{\Lambda}(\nu))$ for all $\nu \in (0,1]$. 

In the finite-framelength regime, the \gls{PLR} is typically characterized by two regions: a waterfall~(WF) region near the DE threshold where the \gls{PLR} decreases sharply, and an error-floor~(EF) region where the \gls{PLR} flattens. In the WF region, according to~\cite{GraelliAmat2018}, the \gls{PLR} can be approximated based on the finite-length scaling of the frame-error rate of LDPC codes~\cite{Amraoui2009}. Specifically, the overall \gls{PLR} $\sum_{\ell = 0}^d  \Lambda_\ell P_\ell$ (averaged over classes) can be approximated by
\vspace{-.05cm}
\begin{align} \label{eq:finite_length_scaling}
	P_{\text{WF}} = P_{G \to 1} Q\bigg(\frac{\sqrt{M} (G^* - \beta(\Lambda) M^{-2/3} - G)}{\sqrt{\alpha^2(\Lambda) + G(1-MG/U)}}\bigg)
\end{align}
where $P_{G \to 1}$ is the \gls{PLR} in the limit $G\to 1$ computed via DE, $Q(\cdot)$ is the Gaussian Q-function, and $\{\alpha(\Lambda),\beta(\Lambda)\}$ are scaling parameters computed as specified in~\cite{Amraoui2006}. 
In the EF region, the \gls{PLR} can be approximated using the method proposed in~\cite{Ivanov2017}. In this region, decoding failures are mainly caused by harmful structures in the corresponding bipartite graph, referred to as stopping sets. A connected bipartite graph $\Sc$ is a stopping
set if all check nodes in $\Sc$ have a degree larger than one. By enumerating the stopping sets, we can approximate $P_\ell$ by
\vspace{-.05cm}
\begin{equation}
	P_{\ell, \text{EF}} = \frac{(U\!-\!1)!}{\Lambda_\ell} \sum_{\Sc \in \Ac}\! \frac{v_\ell(\Sc) c(\Sc) \binom{M}{\psi(\Sc)}\!}{(U \!-\! v(\Sc))!}  \prod_{j=1}^{d} \! \binom{M}{j}^{\!\!-v_j(\Sc)} \! \frac{\Lambda_j^{v_j(\Sc)}}{v_j(\Sc)!\!},
\end{equation}
where $\Ac$ is the set of considered stopping sets, $v(\Sc)$ and $\psi(\Sc)$ are the number of variable nodes and check nodes in $\Sc$, respectively, $v_j(\Sc)$ is the number of degree-$j$ variable nodes in $\Sc$, and $c(\Sc)$ is the number of graphs isomorphic with $\Sc$. 

In~\cite{Munari2020modern}, the \gls{PLR} is approximated as 
\vspace{-.05cm}
\begin{align} \label{eq:PLR_approx_Munari}
	P_\ell \approx P_{\text{WF}} + P_{\ell, \text{EF}}
\end{align}
for the single-class case.\footnote{In the single-class case, this means that the overall PLR is approximated by $P_{\rm WF} + \sum_{\ell = 0}^{d} \Lambda_\ell P_{\ell, {\rm EF}}$. In our paper, it is more convenient to write the approximation in terms of $P_\ell$.} It was shown to be accurate for degree distributions with degrees at least $3$ (see, e.g.,~\cite[Fig.~4]{Munari2020modern}). These degree distributions are typically considered when the design goal is to minimize the EF or maximize the decoding threshold. In our setting, however, it is of interest to consider degree distributions with lower degrees to reduce~$\Phi$. 

In Figs.~\ref{subfig:PLR_deg2} and~\ref{subfig:PLR_deg1}, we investigate the tightness of the approximations~\eqref{eq:PLR_approx_Munari} and $P_\ell \approx P_{\ell, \text{DE}}$ for the setup in Example~\ref{example} and some degree distributions with degrees $1$ and $2$. {Note that an accurate PLR approximation in the WF region is crucial for the computation of the AVP, whereas the AVP is insentitive to low values of the PLR in the EF region where update sporadicity is the dominating factor.} As shown in Fig.~\ref{subfig:PLR_deg2} for the degree distributions $\Lambda\of{1}(x) = \Lambda\of{2}(x) = 0.5x^2 + 0.5x^3$, the approximation~\eqref{eq:PLR_approx_Munari} is loose in the WF region. The reason is that the finite-length scaling leading to~\eqref{eq:finite_length_scaling} is not guaranteed to hold when the bipartite graph contains degree-$2$ variable nodes~\cite{Amraoui2009}. The situation is even worse when degree-$1$ users are present: $P_{\text{WF}}$ is near $P_{G\to 1}$ for all channel load. This makes the approximation~\eqref{eq:PLR_approx_Munari} inaccurate, as shown for the distributions $\{\Lambda\of{1}(x) =\ 0.7x + 0.3x^3$, $\Lambda\of{2}(x) = 0.7x^2 + 0.3x^3\}$ in Fig.~\ref{subfig:PLR_deg1}. On the other hand, $P_{\ell, \text{DE}}$ is an accurate approximation of $P_\ell$ for large values of $\Phi$ corresponding to $G > G^*$, although $P_{\ell, \text{DE}}$ is much lower than $P_\ell$ when $G < G^*$.

In Figs.~\ref{subfig:AVP_deg2} and~\ref{subfig:AVP_deg1}, we show the \gls{AVP} computed with different approximations of the \glspl{PLR}. 
For both sets of degree distributions, setting $P_\ell \!\approx\! P_{\ell, \text{DE}}$ yields an accurate approximation of the \gls{AVP} when $G>G^*$. When class-$1$ users are present as in Fig.~\ref{subfig:AVP_deg1}, we have $G^* = 0$ and the \gls{AVP} approximation obtained by setting $P_\ell \!\approx\! P_{\ell, \text{DE}}$ is accurate for all $G > 0$ (equivalently $\Phi >0$). For moderate values of $G$ corresponding to the WF region, if there are both degree $2$ and higher degrees, the approximation $P_\ell \!\approx\! P_{\ell, \text{DE}}$ leads to an optimistic approximation of the \gls{AVP} while the approximation~\eqref{eq:PLR_approx_Munari} is pessimistic in the WF region, as shown in Fig.~\ref{subfig:AVP_deg2}. In this case, one needs to balance between these two approximations. Our experiments suggest that using $P_\ell \!\approx\! P_{\ell, \text{DE}}$ for $\ell = 2$ and~\eqref{eq:PLR_approx_Munari} for $\ell > 2$, which leads to the dashed-dotted green line in Fig.~\ref{subfig:AVP_deg2}, results in an accurate approximation of the \gls{AVP}.  


From the above observations, we propose the following heuristic \gls{PLR} approximation. For $G>G^*$, we set $P_\ell \approx P_{\ell, \text{DE}}$ for all $\ell$. For $G \le G^*$, we set $P_\ell \approx P_{\ell, \text{DE}}$ for $\ell \le 2$ and $P_\ell \approx P_{\text{WF}} + P_{\ell, \text{EF}}$ for $\ell > 2$.   We summarize the proposed approximation as
	\begin{align} \label{eq:proposed_approx}
		P_\ell \approx \begin{cases}
			P_{\ell, \text{DE}}, &\text{if~$\ell \le 2$ or $G>G^*$}, \\
			P_{\text{WF}} + P_{\ell, \text{EF}}, &\text{if~$\ell > 2$ and $G\le G^*$}. 
		\end{cases}
	\end{align}

\begin{figure} [t!]
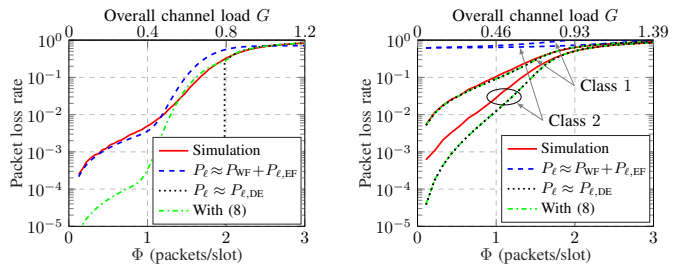
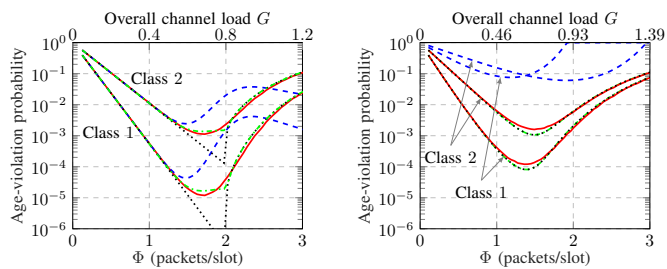

	\centering
	\hspace{-.1cm}
	\subfigure[PLR for $\Lambda\of{1}(x) = \Lambda\of{2}(x) = 0.5x^2 + 0.5x^3$ ($G^* \approx 0.792$)]{\input{fig/PLR_approx_paper_2_new.tex} \hspace{-.1cm} \label{subfig:PLR_deg2}}
	\hspace{.05cm}
	\subfigure[PLR for $\Lambda\of{1}(x) \!= 0.7x \!+ 0.3x^3$, $\Lambda\of{2}(x) = 0.7x^2 \!+\! 0.3x^3$ ($G^* = 0$)]{\input{fig/PLR_approx_paper_1_new.tex} \hspace{-.3cm} \label{subfig:PLR_deg1}}
  \\
	\hspace{-.1cm}
		\subfigure[\gls{AVP} for $\Lambda\of{1}(x) = \Lambda\of{2}(x) = 0.5x^2 + 0.5x^3$ ($G^* \approx 0.792$)]{\input{fig/AVP_approx_paper_2_new.tex} \hspace{-.1cm} \label{subfig:AVP_deg2}}
	\hspace{.05cm}
	\subfigure[\gls{AVP} for $\Lambda\of{1}(x) \!= 0.7x \!+ 0.3x^3$, $\Lambda\of{2}(x) = 0.7x^2 \!+\! 0.3x^3$ ($G^* = 0$)]{\input{fig/AVP_approx_paper_1_new.tex}\hspace{-.2cm} \label{subfig:AVP_deg1}} 
	\vspace{-.3cm}
	\caption{The PLR obtained from simulation or approximation and the corresponding \gls{AVP} vs. $\Phi$ and $G$ for the scenario in Example~\ref{example}. We consider two sets of degree distributions, namely, $\Lambda\of{1}(x) = \Lambda\of{2}(x) = 0.5x^2 + 0.5x^3$ and $\{\Lambda\of{1}(x) = 0.7x + 0.3x^3$, $\Lambda\of{2}(x) = 0.7x^2 + 0.3x^3\}$.}
	\label{fig:PLR_AVP_approx}
	\vspace{-.5cm}
\end{figure}

\vspace{-.1cm}
\section{Numerical Results}
\vspace{-.1cm}
We next solve the optimization~\eqref{eq:power_minimization} for the scenario in Example~\ref{example} with different target~\glspl{AVP}. We let $d = 3$, i.e., $\Lambda\of{k}(x) = \sum_{\ell=0}^{3}\Lambda\of{k}_\ell x^\ell$. We keep the same update probability $\mu_1 \!=\! \mu_2 \!=\! \mu$ for both classes and control the relative difference between the probabilities of activating users in different classes via $(\Lambda\of{1}_0,\Lambda\of{2}_0)$. As $\Lambda\of{k}_3 = 1 - \sum_{\ell = 0}^{2}  \Lambda\of{k}_\ell$, the optimization variables are $\big(\mu, \{\Lambda\of{k}_\ell\}_{k\in \{1,2\}, \ell \in \{0,1,2\}}\big)$. 
The \gls{AVP} is computed as in Proposition~\ref{prop:peak_age_violation} with the approximated/simulated PLR.  
We numerically solve~\eqref{eq:power_minimization} by means of the Nelder-Mead simplex algorithm~\cite{nelder1965simplex}, a commonly-used search method for multidimensional nonlinear optimization. However, we note that this heuristic method can converge to nonstationary points and is highly sensitive to the initial values of $\{\mu,\Lambda\of{k}\}$. We try multiple initializations by sampling the search space with a step $0.1$, and by running the optimization multiple times.

\begin{table*}[t!]
	\caption{Optimized update probability and degree distributions for Example~\ref{example} with approximate \gls{PLR} and different target AVPs $(\epsilon_1,\epsilon_2)$}
	\label{tab:optimized_dist}
	\footnotesize
	\centering
		\renewcommand{\arraystretch}{0.9}
	\vspace{-.2cm}
	\begin{tabular}{c  c  c  c  c  c  c  c  c  c  c  c  c } 
		\toprule 
		$\epsilon_1$ & $\epsilon_2$ & $~U \mu~$ & $~\Lambda\of{1}_0~$ & $~\Lambda\of{1}_1~$ & $~\Lambda\of{1}_2~$ & $~\Lambda\of{1}_3~$ & $~\Lambda\of{2}_0~$ & $~\Lambda\of{2}_1~$ & $~\Lambda\of{2}_2~$ & $~\Lambda\of{2}_3~$ & $\Phi$ & $[1\!-\!(1\!-\!\mu)^M](1\!-\!\Lambda_0)$ \\
		\midrule 
		$10^{-5}$ & $10^{-3}$ & $0.71$ & $0.05$ & $0$ & $0.24$ &   $0.71$  &  $0.05$ & $0$  &  $0.25$  &  $0.7$ & $1.77$ & $0.017$ \\
		$10^{-4}$ & $10^{-3}$ & $0.65$ & $0.19$ & $0$ & $0.11$ &   $0.7$  &  $0.02$ & $0$  &  $0.46$  &  $0.52$ & $1.59$ & $0.015$ \\
		$10^{-4}$ & $10^{-2}$ & $0.69$ & $0.25$ & $0.19$ & $0.56$ &   $0$  &  $0.39$ & $0.01$  &  $0.6$  &  $0$ & $0.87$ & $0.011$ \\
		$10^{-3}$ & $10^{-2}$ & $0.56$ & $0.24$ & $0.36$ & $0.4$ &   $0$  &  $0.21$ & $0.12$  &  $0.67$  &  $0$ & $0.79$ & $0.011$ \\
		$10^{-3}$ & $10^{-1}$ & $0.59$ & $0.13$ & $0.87$ & $0$ &   $0$  &  $0.52$ & $0.48$  &  $0$  &  $0$ & $0.33$ & $0.008$ \\
		\bottomrule
	\end{tabular}
	\vspace{-.3cm}
\end{table*}

The optimization results with approximate \gls{PLR} for some different target \glspl{AVP} $(\epsilon_1,\epsilon_2)$ are shown in Table~\ref{tab:optimized_dist}. We observe that for the mild requirement $(\epsilon_1,\epsilon_2) = (10^{-3},10^{-1})$, using degrees $0$ and $1$ is sufficient. As the requirement becomes more stringent, one needs to use an increasing fraction of degree~$2$ and eventually degree~$3$. Also, the users should be activated more frequently, indicated by the probability $[1-(1-\mu)^M](1-\Lambda_0)$ shown in the last column of Table~\ref{tab:optimized_dist}. As compared to the regular distributions in Fig.~\ref{fig:AVP_vs_power}, our optimized irregular degree distributions reduce $\Phi$ by about $8\%$ and $20\%$ for the requirements $(\epsilon_1,\epsilon_2) = (10^{-5},10^{-3})$ and $(\epsilon_1,\epsilon_2) = (10^{-4},10^{-2})$, respectively.

In~Fig.~\ref{fig:AVPopt_vs_power}, we show the AVP of the optimized distributions using approximate \gls{PLR}~\eqref{eq:proposed_approx} shown in Table~\ref{tab:optimized_dist} and compare it with the optimized distributions using simulated \gls{PLR}. The latter can further reduce $\Phi$ by no more than $0.04$ packets/slot. This confirms that our proposed PLR approximation is sufficiently accurate. Further experiments also show that using degrees higher than~$3$ is not beneficial for the considered parameters. 

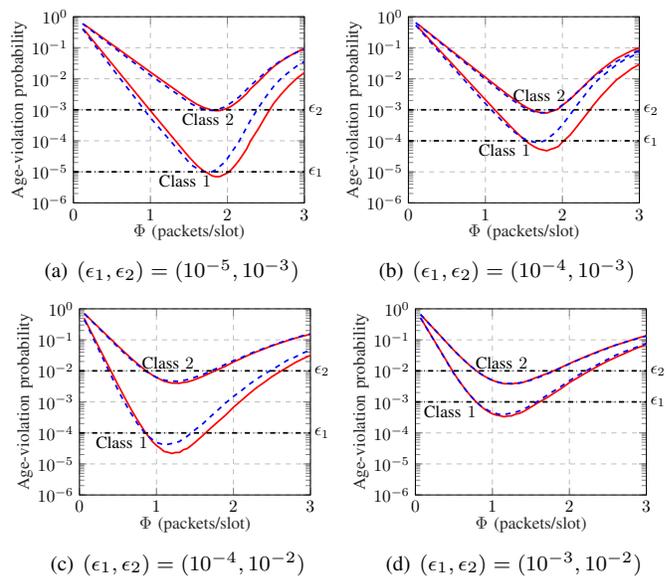
\begin{figure} [t!]
	\centering
	\hspace{-.25cm}
	\subfigure[$(\epsilon_1,\epsilon_2) = (10^{-5},10^{-3})$]{
%
%
\begin{tikzpicture}[scale=.65]

\begin{axis}[%
width=1.86in,
height=1.5in,
at={(0.553in,0.43in)},
scale only axis,
unbounded coords=jump,
xmin=0,
xmax=3,
xlabel style={font=\color{white!15!black},yshift=4pt},
xlabel={$\Phi$ (packets/slot)},
ymode=log,
ymin=1e-06,
ymax=1,
ytick={ 1e-6, 1e-05, 0.0001,  0.001,   0.01,    0.1,      1},
yminorticks=true,
ylabel style={font=\color{white!15!black},yshift=-4pt},
ylabel={{Age-violation probability}},
axis background/.style={fill=white},
title style={font=\bfseries},
xmajorgrids,
ymajorgrids,
legend style={at={(0.411,0.146)}, anchor=south west, legend cell align=left, align=left, draw=white!15!black},
clip=false
]

\addplot [line width = 1,color=red, forget plot]
table[row sep=crcr]{%
1.32538849e-01 3.98045563e-01 \\ 
5.00000000e-01 3.07944606e-02 \\ 
1.00000000e+00 9.36489825e-04 \\ 
1.50000000e+00 3.04615346e-05 \\ 
1.55000000e+00 2.24866472e-05 \\ 
1.60000000e+00 1.65308583e-05 \\ 
1.65000000e+00 1.29752393e-05 \\ 
1.70000000e+00 1.03564145e-05 \\ 
1.75000000e+00 8.96442933e-06 \\ 
1.80000000e+00 7.45053961e-06 \\ 
1.85000000e+00 6.98819660e-06 \\ 
1.90000000e+00 7.01571957e-06 \\ 
1.95000000e+00 8.03174625e-06 \\ 
2.00000000e+00 9.15174655e-06 \\ 
2.05000000e+00 1.17821667e-05 \\ 
2.10000000e+00 1.62508635e-05 \\ 
2.15000000e+00 2.33537602e-05 \\ 
2.20000000e+00 3.54612447e-05 \\ 
2.25000000e+00 5.48815675e-05 \\ 
2.30000000e+00 9.92344244e-05 \\ 
2.40000000e+00 2.39245385e-04 \\ 
2.50000000e+00 6.38205034e-04 \\ 
2.60000000e+00 1.59004395e-03 \\ 
2.70000000e+00 3.22075579e-03 \\ 
2.80000000e+00 5.95451720e-03 \\ 
2.90000000e+00 1.02738954e-02 \\ 
3.00000000e+00 1.57833974e-02 \\ 
};
\addplot [line width = 1,color=red]
table[row sep=crcr]{%
1.32538849e-01 5.80894816e-01 \\ 
5.00000000e-01 1.28467993e-01 \\ 
1.00000000e+00 1.63955544e-02 \\ 
1.50000000e+00 2.18353479e-03 \\ 
1.55000000e+00 1.82838971e-03 \\ 
1.60000000e+00 1.52689065e-03 \\ 
1.65000000e+00 1.32666508e-03 \\ 
1.70000000e+00 1.16480949e-03 \\ 
1.75000000e+00 1.07374466e-03 \\ 
1.80000000e+00 9.65375882e-04 \\ 
1.85000000e+00 9.34096585e-04 \\ 
1.90000000e+00 9.41128110e-04 \\ 
1.95000000e+00 1.02669132e-03 \\ 
2.00000000e+00 1.11621859e-03 \\ 
2.05000000e+00 1.30032283e-03 \\ 
2.10000000e+00 1.58862127e-03 \\ 
2.15000000e+00 1.98048831e-03 \\ 
2.20000000e+00 2.55009947e-03 \\ 
2.25000000e+00 3.31679355e-03 \\ 
2.30000000e+00 4.72845854e-03 \\ 
2.40000000e+00 8.03305167e-03 \\ 
2.50000000e+00 1.43431065e-02 \\ 
2.60000000e+00 2.45315292e-02 \\ 
2.70000000e+00 3.71850864e-02 \\ 
2.80000000e+00 5.31805066e-02 \\ 
2.90000000e+00 7.28708749e-02 \\ 
3.00000000e+00 9.34221254e-02 \\ 
};

\addplot [line width = 1,color=blue, dashed, forget plot]
table[row sep=crcr]{%
	 1.17699517e-01 4.07739745e-01 \\ 
	5.00000000e-01 2.25711332e-02 \\ 
	1.00000000e+00 5.51128785e-04 \\ 
	1.50000000e+00 2.04133550e-05 \\ 
	1.55000000e+00 1.55804460e-05 \\ 
	1.60000000e+00 1.30736503e-05 \\ 
	1.65000000e+00 1.13630368e-05 \\ 
	1.70000000e+00 1.02953285e-05 \\ 
	1.75000000e+00 9.47795361e-06 \\ 
	1.80000000e+00 1.00630774e-05 \\ 
	1.85000000e+00 1.16682973e-05 \\ 
	1.90000000e+00 1.49440608e-05 \\ 
	1.95000000e+00 1.91223253e-05 \\ 
	2.00000000e+00 2.69198825e-05 \\ 
	2.05000000e+00 4.09726145e-05 \\ 
	2.10000000e+00 6.22551527e-05 \\ 
	2.15000000e+00 9.45984997e-05 \\ 
	2.20000000e+00 1.60606647e-04 \\ 
	2.25000000e+00 2.73003746e-04 \\ 
	2.30000000e+00 4.57739359e-04 \\ 
	2.40000000e+00 1.13713525e-03 \\ 
	2.50000000e+00 2.64261580e-03 \\ 
	2.60000000e+00 5.66585129e-03 \\ 
	2.70000000e+00 1.04064478e-02 \\ 
	2.80000000e+00 1.65233161e-02 \\ 
	2.90000000e+00 2.56522007e-02 \\ 
	3.00000000e+00 3.59019302e-02 \\ 
};
\addplot [line width = 1,color=blue,dashed]
table[row sep=crcr]{%
	 1.17699517e-01 5.99386832e-01 \\ 
	5.00000000e-01 1.13832312e-01 \\ 
	1.00000000e+00 1.31202251e-02 \\ 
	1.50000000e+00 1.74631167e-03 \\ 
	1.55000000e+00 1.46819178e-03 \\ 
	1.60000000e+00 1.28288449e-03 \\ 
	1.65000000e+00 1.14372497e-03 \\ 
	1.70000000e+00 1.04440925e-03 \\ 
	1.75000000e+00 9.65751638e-04 \\ 
	1.80000000e+00 9.58351078e-04 \\ 
	1.85000000e+00 9.89966030e-04 \\ 
	1.90000000e+00 1.08825284e-03 \\ 
	1.95000000e+00 1.19118994e-03 \\ 
	2.00000000e+00 1.39960120e-03 \\ 
	2.05000000e+00 1.70242577e-03 \\ 
	2.10000000e+00 2.07613840e-03 \\ 
	2.15000000e+00 2.56758008e-03 \\ 
	2.20000000e+00 3.37578326e-03 \\ 
	2.25000000e+00 4.51674729e-03 \\ 
	2.30000000e+00 6.00475028e-03 \\ 
	2.40000000e+00 1.00422259e-02 \\ 
	2.50000000e+00 1.67186951e-02 \\ 
	2.60000000e+00 2.65795210e-02 \\ 
	2.70000000e+00 3.92461411e-02 \\ 
	2.80000000e+00 5.31201296e-02 \\ 
	2.90000000e+00 7.11922519e-02 \\ 
	3.00000000e+00 9.06916052e-02 \\ 
};

\addplot [color=black, dashdotted, line width = 1, forget plot]
table[row sep=crcr]{%
	0	0.00001\\
	3	0.00001\\
};
\addplot [color=black, dashdotted, line width = 1]
table[row sep=crcr]{%
	0	0.001\\
	3	0.001\\
};

\node at (axis cs:1.45, 5e-6)  (class1)    {Class $1$};
\node at (axis cs:1.75, 5e-4)  (class2)    {Class $2$};

\node at (axis cs:3.15, .00001)  ()    {$\epsilon_1$};
\node at (axis cs:3.15, .001)  ()    {$\epsilon_2$};

\end{axis}
\end{tikzpicture}%
 \label{subfig:AVPopt1}}
	\hspace{-.2cm}
	\subfigure[$(\epsilon_1,\epsilon_2) = (10^{-4},10^{-3})$]{
%
%
\begin{tikzpicture}[scale=.65]
	
	\begin{axis}[%
		width=1.86in,
		height=1.5in,
		at={(0.553in,0.43in)},
		scale only axis,
		unbounded coords=jump,
		xmin=0,
		xmax=3,
		xlabel style={font=\color{white!15!black},yshift=4pt},
		xlabel={$\Phi$ (packets/slot)},
		ymode=log,
		ymin=1e-06,
		ymax=1,
		ytick={ 1e-6, 1e-05, 0.0001,  0.001,   0.01,    0.1,      1},
		yminorticks=true,
		ylabel style={font=\color{white!15!black},yshift=-4pt},
		ylabel={{Age-violation probability}},
		axis background/.style={fill=white},
		title style={font=\bfseries},
		xmajorgrids,
		ymajorgrids,
		legend style={at={(0.411,0.146)}, anchor=south west, legend cell align=left, align=left, draw=white!15!black},
		clip=false
		]
		
		\addplot [line width = 1,color=red, forget plot]
		table[row sep=crcr]{%
			 1.22461496e-01 4.71279078e-01 \\ 
			5.00000000e-01 4.61696233e-02 \\ 
			1.00000000e+00 2.12006563e-03 \\ 
			1.05000000e+00 1.55856420e-03 \\ 
			1.10000000e+00 1.14782761e-03 \\ 
			1.15000000e+00 8.45004504e-04 \\ 
			1.20000000e+00 6.24342211e-04 \\ 
			1.25000000e+00 4.61231577e-04 \\ 
			1.30000000e+00 3.41792386e-04 \\ 
			1.35000000e+00 2.55897398e-04 \\ 
			1.40000000e+00 1.92542507e-04 \\ 
			1.45000000e+00 1.44797107e-04 \\ 
			1.50000000e+00 1.13880412e-04 \\ 
			1.55000000e+00 8.82028854e-05 \\ 
			1.60000000e+00 7.27110408e-05 \\ 
			1.65000000e+00 6.13949573e-05 \\ 
			1.70000000e+00 5.41539973e-05 \\ 
			1.75000000e+00 5.07116888e-05 \\ 
			1.80000000e+00 4.77843781e-05 \\ 
			1.85000000e+00 5.37028103e-05 \\ 
			1.90000000e+00 5.83710145e-05 \\ 
			1.95000000e+00 6.63861661e-05 \\ 
			2.00000000e+00 9.10571400e-05 \\ 
			2.10000000e+00 1.53398989e-04 \\ 
			2.20000000e+00 3.13160676e-04 \\ 
			2.30000000e+00 6.48802375e-04 \\ 
			2.40000000e+00 1.37989097e-03 \\ 
			2.50000000e+00 2.73884995e-03 \\ 
			2.60000000e+00 5.01494366e-03 \\ 
			2.70000000e+00 9.13766018e-03 \\ 
			2.80000000e+00 1.38223962e-02 \\ 
			2.90000000e+00 2.12860403e-02 \\ 
			3.00000000e+00 2.89666482e-02 \\ 
		};
		\addplot [line width = 1,color=red]
		table[row sep=crcr]{%
			 1.22461496e-01 5.77293328e-01 \\ 
			5.00000000e-01 1.05827078e-01 \\ 
			1.00000000e+00 1.11897650e-02 \\ 
			1.05000000e+00 8.94537329e-03 \\ 
			1.10000000e+00 7.16059416e-03 \\ 
			1.15000000e+00 5.72556190e-03 \\ 
			1.20000000e+00 4.60502432e-03 \\ 
			1.25000000e+00 3.69432323e-03 \\ 
			1.30000000e+00 2.97514013e-03 \\ 
			1.35000000e+00 2.41314121e-03 \\ 
			1.40000000e+00 1.96878372e-03 \\ 
			1.45000000e+00 1.60925726e-03 \\ 
			1.50000000e+00 1.35802914e-03 \\ 
			1.55000000e+00 1.13669750e-03 \\ 
			1.60000000e+00 9.94935346e-04 \\ 
			1.65000000e+00 8.95354111e-04 \\ 
			1.70000000e+00 8.27158943e-04 \\ 
			1.75000000e+00 8.02604998e-04 \\ 
			1.80000000e+00 7.86543515e-04 \\ 
			1.85000000e+00 8.79494498e-04 \\ 
			1.90000000e+00 9.55165012e-04 \\ 
			1.95000000e+00 1.08118751e-03 \\ 
			2.00000000e+00 1.39978820e-03 \\ 
			2.10000000e+00 2.16923093e-03 \\ 
			2.20000000e+00 3.78368283e-03 \\ 
			2.30000000e+00 6.63986753e-03 \\ 
			2.40000000e+00 1.17243917e-02 \\ 
			2.50000000e+00 1.91403766e-02 \\ 
			2.60000000e+00 2.95647180e-02 \\ 
			2.70000000e+00 4.48507344e-02 \\ 
			2.80000000e+00 5.99869974e-02 \\ 
			2.90000000e+00 8.01681625e-02 \\ 
			3.00000000e+00 9.87772367e-02 \\ 
		};
		
		\addplot [line width = 1,color=blue, dashed, forget plot]
		table[row sep=crcr]{%
			9.49006434e-02 5.36025073e-01 \\ 
			5.00000000e-01 3.81404967e-02 \\ 
			1.00000000e+00 1.59934908e-03 \\ 
			1.05000000e+00 1.18719047e-03 \\ 
			1.10000000e+00 8.68013535e-04 \\ 
			1.15000000e+00 6.52186713e-04 \\ 
			1.20000000e+00 4.83555509e-04 \\ 
			1.25000000e+00 3.66429081e-04 \\ 
			1.30000000e+00 2.86007850e-04 \\ 
			1.35000000e+00 2.17178885e-04 \\ 
			1.40000000e+00 1.74034595e-04 \\ 
			1.45000000e+00 1.41971895e-04 \\ 
			1.50000000e+00 1.17365739e-04 \\ 
			1.55000000e+00 1.04037983e-04 \\ 
			1.60000000e+00 9.43407762e-05 \\ 
			1.65000000e+00 8.96848713e-05 \\ 
			1.70000000e+00 8.92046880e-05 \\ 
			1.75000000e+00 9.68454336e-05 \\ 
			1.80000000e+00 1.10066906e-04 \\ 
			1.85000000e+00 1.42519190e-04 \\ 
			1.90000000e+00 1.76788757e-04 \\ 
			1.95000000e+00 2.37078012e-04 \\ 
			2.00000000e+00 3.28074259e-04 \\ 
			2.10000000e+00 6.93008441e-04 \\ 
			2.20000000e+00 1.53228689e-03 \\ 
			2.30000000e+00 3.16913507e-03 \\ 
			2.40000000e+00 6.60572721e-03 \\ 
			2.50000000e+00 1.11844265e-02 \\ 
			2.60000000e+00 1.76658783e-02 \\ 
			2.70000000e+00 2.80968642e-02 \\ 
			2.80000000e+00 4.00382535e-02 \\ 
			2.90000000e+00 5.29467715e-02 \\ 
			3.00000000e+00 7.18131784e-02 \\ 
		};
		\addplot [line width = 1,color=blue,dashed]
		table[row sep=crcr]{%
			9.49006434e-02 6.43784034e-01 \\ 
			5.00000000e-01 9.80419987e-02 \\ 
			1.00000000e+00 9.68567600e-03 \\ 
			1.05000000e+00 7.70560970e-03 \\ 
			1.10000000e+00 6.12640560e-03 \\ 
			1.15000000e+00 4.89207594e-03 \\ 
			1.20000000e+00 3.90987160e-03 \\ 
			1.25000000e+00 3.15005348e-03 \\ 
			1.30000000e+00 2.55990651e-03 \\ 
			1.35000000e+00 2.06970361e-03 \\ 
			1.40000000e+00 1.71111494e-03 \\ 
			1.45000000e+00 1.41439950e-03 \\ 
			1.50000000e+00 1.19112029e-03 \\ 
			1.55000000e+00 1.03992168e-03 \\ 
			1.60000000e+00 9.21447451e-04 \\ 
			1.65000000e+00 8.44023542e-04 \\ 
			1.70000000e+00 7.96563960e-04 \\ 
			1.75000000e+00 8.03758765e-04 \\ 
			1.80000000e+00 8.29377327e-04 \\ 
			1.85000000e+00 9.16475640e-04 \\ 
			1.90000000e+00 1.00639274e-03 \\ 
			1.95000000e+00 1.17115043e-03 \\ 
			2.00000000e+00 1.42131355e-03 \\ 
			2.10000000e+00 2.27104870e-03 \\ 
			2.20000000e+00 3.82274293e-03 \\ 
			2.30000000e+00 6.35438000e-03 \\ 
			2.40000000e+00 1.10392863e-02 \\ 
			2.50000000e+00 1.66852207e-02 \\ 
			2.60000000e+00 2.41724759e-02 \\ 
			2.70000000e+00 3.58834933e-02 \\ 
			2.80000000e+00 4.88762497e-02 \\ 
			2.90000000e+00 6.31075420e-02 \\ 
			3.00000000e+00 8.26591286e-02 \\ 
		};
		
		\addplot [color=black, dashdotted, line width = 1, forget plot]
		table[row sep=crcr]{%
			0	0.0001\\
			3	0.0001\\
		};
		\addplot [color=black, dashdotted, line width = 1]
		table[row sep=crcr]{%
			0	0.001\\
			3	0.001\\
		};
		
		\node at (axis cs:1.25, 5e-5)  (class1)    {Class $1$};
		\node at (axis cs:1.7, .003)  (class2)    {Class $2$};
		
		\node at (axis cs:3.15, .0001)  ()    {$\epsilon_1$};
		\node at (axis cs:3.15, .001)  ()    {$\epsilon_2$};
		
	\end{axis}
\end{tikzpicture}%
 \label{subfig:AVPopt2}} \hspace{-.1cm} \\
	\hspace{-.2cm}
	\subfigure[$(\epsilon_1,\epsilon_2) = (10^{-4},10^{-2})$]{
%
%
\begin{tikzpicture}[scale=.65]
	
	\begin{axis}[%
		width=1.86in,
		height=1.5in,
		at={(0.553in,0.43in)},
		scale only axis,
		unbounded coords=jump,
		xmin=0,
		xmax=3,
		xlabel style={font=\color{white!15!black},yshift=4pt},
		xlabel={$\Phi$ (packets/slot)},
		ymode=log,
		ymin=1e-06,
		ymax=1,
		ytick={ 1e-6, 1e-05, 0.0001,  0.001,   0.01,    0.1,      1},
		yminorticks=true,
		ylabel style={font=\color{white!15!black},yshift=-4pt},
		ylabel={{Age-violation probability}},
		axis background/.style={fill=white},
		title style={font=\bfseries},
		xmajorgrids,
		ymajorgrids,
		legend style={at={(0.411,0.146)}, anchor=south west, legend cell align=left, align=left, draw=white!15!black},
		clip=false
		]
		
		\addplot [line width = 1,color=red, forget plot]
		table[row sep=crcr]{%
			 6.15032144e-02 4.95324666e-01 \\ 
			5.00000000e-01 3.54401074e-03 \\ 
			5.50000000e-01 2.07997725e-03 \\ 
			6.00000000e-01 1.20553268e-03 \\ 
			6.50000000e-01 7.15436641e-04 \\ 
			7.00000000e-01 4.33010191e-04 \\ 
			7.50000000e-01 2.68006504e-04 \\ 
			8.00000000e-01 1.71087385e-04 \\ 
			8.50000000e-01 1.09361672e-04 \\ 
			9.00000000e-01 7.47761510e-05 \\ 
			9.50000000e-01 5.19512256e-05 \\ 
			1.00000000e+00 3.95018106e-05 \\ 
			1.05000000e+00 3.19739482e-05 \\ 
			1.10000000e+00 2.52281363e-05 \\ 
			1.15000000e+00 2.37798526e-05 \\ 
			1.20000000e+00 2.17895779e-05 \\ 
			1.25000000e+00 2.33976994e-05 \\ 
			1.30000000e+00 2.34111560e-05 \\ 
			1.35000000e+00 2.86906910e-05 \\ 
			1.40000000e+00 3.13249071e-05 \\ 
			1.45000000e+00 3.89843280e-05 \\ 
			1.50000000e+00 4.93296369e-05 \\ 
			1.60000000e+00 8.16463761e-05 \\ 
			1.70000000e+00 1.39167966e-04 \\ 
			1.80000000e+00 2.30950765e-04 \\ 
			1.90000000e+00 3.77619759e-04 \\ 
			2.00000000e+00 6.72085612e-04 \\ 
			2.10000000e+00 1.09630295e-03 \\ 
			2.20000000e+00 1.79843250e-03 \\ 
			2.30000000e+00 2.86284641e-03 \\ 
			2.40000000e+00 4.27259238e-03 \\ 
			2.50000000e+00 6.51002524e-03 \\ 
			2.60000000e+00 9.28240815e-03 \\ 
			2.70000000e+00 1.30157368e-02 \\ 
			2.80000000e+00 1.84051719e-02 \\ 
			2.90000000e+00 2.43701657e-02 \\ 
			3.00000000e+00 3.17278241e-02 \\ 
		};
		\addplot [line width = 1,color=red]
		table[row sep=crcr]{%
			 6.15032144e-02 7.10467933e-01 \\ 
			5.00000000e-01 6.33858844e-02 \\ 
			5.50000000e-01 4.85494928e-02 \\ 
			6.00000000e-01 3.70616949e-02 \\ 
			6.50000000e-01 2.85351724e-02 \\ 
			7.00000000e-01 2.21197819e-02 \\ 
			7.50000000e-01 1.72631034e-02 \\ 
			8.00000000e-01 1.36390671e-02 \\ 
			8.50000000e-01 1.07706007e-02 \\ 
			9.00000000e-01 8.80876769e-03 \\ 
			9.50000000e-01 7.21495010e-03 \\ 
			1.00000000e+00 6.11817583e-03 \\ 
			1.05000000e+00 5.34602401e-03 \\ 
			1.10000000e+00 4.62595246e-03 \\ 
			1.15000000e+00 4.32354899e-03 \\ 
			1.20000000e+00 4.04002618e-03 \\ 
			1.25000000e+00 4.03703583e-03 \\ 
			1.30000000e+00 3.97734216e-03 \\ 
			1.35000000e+00 4.26161261e-03 \\ 
			1.40000000e+00 4.36301558e-03 \\ 
			1.45000000e+00 4.78269742e-03 \\ 
			1.50000000e+00 5.24723576e-03 \\ 
			1.60000000e+00 6.63035523e-03 \\ 
			1.70000000e+00 8.54631074e-03 \\ 
			1.80000000e+00 1.11157532e-02 \\ 
			1.90000000e+00 1.42472142e-02 \\ 
			2.00000000e+00 1.92961129e-02 \\ 
			2.10000000e+00 2.49145293e-02 \\ 
			2.20000000e+00 3.25077761e-02 \\ 
			2.30000000e+00 4.12934109e-02 \\ 
			2.40000000e+00 5.11459950e-02 \\ 
			2.50000000e+00 6.44160234e-02 \\ 
			2.60000000e+00 7.77742668e-02 \\ 
			2.70000000e+00 9.38963383e-02 \\ 
			2.80000000e+00 1.12199031e-01 \\ 
			2.90000000e+00 1.30837504e-01 \\ 
			3.00000000e+00 1.51252125e-01 \\ 
		};
		
		\addplot [line width = 1,color=blue, dashed, forget plot]
		table[row sep=crcr]{%
			6.58297256e-02 4.37876916e-01 \\ 
			5.00000000e-01 2.43425340e-03 \\ 
			5.50000000e-01 1.41297350e-03 \\ 
			6.00000000e-01 8.51881310e-04 \\ 
			6.50000000e-01 5.05529861e-04 \\ 
			7.00000000e-01 3.28563710e-04 \\ 
			7.50000000e-01 2.04525722e-04 \\ 
			8.00000000e-01 1.45582533e-04 \\ 
			8.50000000e-01 9.73050624e-05 \\ 
			9.00000000e-01 7.63259090e-05 \\ 
			9.50000000e-01 6.02732509e-05 \\ 
			1.00000000e+00 4.91492461e-05 \\ 
			1.05000000e+00 4.52140859e-05 \\ 
			1.10000000e+00 4.12060330e-05 \\ 
			1.15000000e+00 4.29270464e-05 \\ 
			1.20000000e+00 4.56299902e-05 \\ 
			1.25000000e+00 4.90087295e-05 \\ 
			1.30000000e+00 5.30135575e-05 \\ 
			1.35000000e+00 6.59544882e-05 \\ 
			1.40000000e+00 7.82545156e-05 \\ 
			1.45000000e+00 9.83046626e-05 \\ 
			1.50000000e+00 1.25170565e-04 \\ 
			1.60000000e+00 1.95073436e-04 \\ 
			1.70000000e+00 3.31802279e-04 \\ 
			1.80000000e+00 5.53969340e-04 \\ 
			1.90000000e+00 9.30281340e-04 \\ 
			2.00000000e+00 1.50066345e-03 \\ 
			2.10000000e+00 2.27013989e-03 \\ 
			2.20000000e+00 3.60536770e-03 \\ 
			2.30000000e+00 5.43480141e-03 \\ 
			2.40000000e+00 7.80856288e-03 \\ 
			2.50000000e+00 1.09259265e-02 \\ 
			2.60000000e+00 1.53109770e-02 \\ 
			2.70000000e+00 2.19826453e-02 \\ 
			2.80000000e+00 2.78281097e-02 \\ 
			2.90000000e+00 3.72543773e-02 \\ 
			3.00000000e+00 4.72187080e-02 \\ 
		};
		\addplot [line width = 1,color=blue,dashed]
		table[row sep=crcr]{%
			6.58297256e-02 6.84226829e-01 \\ 
			5.00000000e-01 5.81506430e-02 \\ 
			5.50000000e-01 4.43077743e-02 \\ 
			6.00000000e-01 3.39399153e-02 \\ 
			6.50000000e-01 2.61036852e-02 \\ 
			7.00000000e-01 2.03911623e-02 \\ 
			7.50000000e-01 1.59579868e-02 \\ 
			8.00000000e-01 1.27983977e-02 \\ 
			8.50000000e-01 1.02392454e-02 \\ 
			9.00000000e-01 8.57569574e-03 \\ 
			9.50000000e-01 7.17131916e-03 \\ 
			1.00000000e+00 6.20478084e-03 \\ 
			1.05000000e+00 5.59913366e-03 \\ 
			1.10000000e+00 5.03871211e-03 \\ 
			1.15000000e+00 4.80226540e-03 \\ 
			1.20000000e+00 4.68043646e-03 \\ 
			1.25000000e+00 4.62953750e-03 \\ 
			1.30000000e+00 4.58903552e-03 \\ 
			1.35000000e+00 4.92601973e-03 \\ 
			1.40000000e+00 5.19717236e-03 \\ 
			1.45000000e+00 5.60535557e-03 \\ 
			1.50000000e+00 6.23946727e-03 \\ 
			1.60000000e+00 7.54380765e-03 \\ 
			1.70000000e+00 9.75985844e-03 \\ 
			1.80000000e+00 1.25401901e-02 \\ 
			1.90000000e+00 1.64886886e-02 \\ 
			2.00000000e+00 2.13581816e-02 \\ 
			2.10000000e+00 2.65247931e-02 \\ 
			2.20000000e+00 3.48046289e-02 \\ 
			2.30000000e+00 4.34160791e-02 \\ 
			2.40000000e+00 5.36566423e-02 \\ 
			2.50000000e+00 6.56542301e-02 \\ 
			2.60000000e+00 8.04765048e-02 \\ 
			2.70000000e+00 9.77882670e-02 \\ 
			2.80000000e+00 1.14166310e-01 \\ 
			2.90000000e+00 1.36831596e-01 \\ 
			3.00000000e+00 1.55185175e-01 \\ 
		};
		
		\addplot [color=black, dashdotted, line width = 1, forget plot]
		table[row sep=crcr]{%
			0	0.0001\\
			3	0.0001\\
		};
		\addplot [color=black, dashdotted, line width = 1]
		table[row sep=crcr]{%
			0	0.01\\
			3	0.01\\
		};
		
		\node at (axis cs:0.55, 5e-5)  (class1)    {Class $1$};
		\node at (axis cs:1.15, .018)  (class2)    {Class $2$};
		
		\node at (axis cs:3.15, .0001)  ()    {$\epsilon_1$};
		\node at (axis cs:3.15, .01)  ()    {$\epsilon_2$};
		
	\end{axis}
\end{tikzpicture}%
 \label{subfig:AVPopt3}}
	\hspace{-.2cm}
	\subfigure[$(\epsilon_1,\epsilon_2) = (10^{-3},10^{-2})$]{
%
%
\begin{tikzpicture}[scale=.65]
	
	\begin{axis}[%
		width=1.86in,
		height=1.5in,
		at={(0.553in,0.43in)},
		scale only axis,
		unbounded coords=jump,
		xmin=0,
		xmax=3,
		xlabel style={font=\color{white!15!black},yshift=4pt},
		xlabel={$\Phi$ (packets/slot)},
		ymode=log,
		ymin=1e-06,
		ymax=1,
		ytick={ 1e-6, 1e-05, 0.0001,  0.001,   0.01,    0.1,      1},
		yminorticks=true,
		ylabel style={font=\color{white!15!black},yshift=-4pt},
		ylabel={{Age-violation probability}},
		axis background/.style={fill=white},
		title style={font=\bfseries},
		xmajorgrids,
		ymajorgrids,
		legend style={at={(0.411,0.146)}, anchor=south west, legend cell align=left, align=left, draw=white!15!black},
		clip=false
		]
		
		\addplot [line width = 1,color=red, forget plot]
		table[row sep=crcr]{%
			 6.99699868e-02 4.94993079e-01 \\ 
			5.00000000e-01 8.44317557e-03 \\ 
			5.50000000e-01 5.48814195e-03 \\ 
			6.00000000e-01 3.68231753e-03 \\ 
			6.50000000e-01 2.51593418e-03 \\ 
			7.00000000e-01 1.75615842e-03 \\ 
			7.50000000e-01 1.25902405e-03 \\ 
			8.00000000e-01 9.40682155e-04 \\ 
			8.50000000e-01 7.25798478e-04 \\ 
			9.00000000e-01 5.78813019e-04 \\ 
			9.50000000e-01 4.81742950e-04 \\ 
			1.00000000e+00 4.13481589e-04 \\ 
			1.05000000e+00 3.76241739e-04 \\ 
			1.10000000e+00 3.57979572e-04 \\ 
			1.15000000e+00 3.31899780e-04 \\ 
			1.20000000e+00 3.50931555e-04 \\ 
			1.25000000e+00 3.54387844e-04 \\ 
			1.30000000e+00 4.00783462e-04 \\ 
			1.35000000e+00 4.25781705e-04 \\ 
			1.40000000e+00 4.82627040e-04 \\ 
			1.45000000e+00 5.70740799e-04 \\ 
			1.50000000e+00 6.78223132e-04 \\ 
			1.60000000e+00 9.65508773e-04 \\ 
			1.70000000e+00 1.32499546e-03 \\ 
			1.80000000e+00 1.91424413e-03 \\ 
			1.90000000e+00 2.77211487e-03 \\ 
			2.00000000e+00 4.17846020e-03 \\ 
			2.10000000e+00 5.83421561e-03 \\ 
			2.20000000e+00 8.20210205e-03 \\ 
			2.30000000e+00 1.13263956e-02 \\ 
			2.40000000e+00 1.55008581e-02 \\ 
			2.50000000e+00 2.07833253e-02 \\ 
			2.60000000e+00 2.71322861e-02 \\ 
			2.70000000e+00 3.48912088e-02 \\ 
			2.80000000e+00 4.51552204e-02 \\ 
			2.90000000e+00 5.71205651e-02 \\ 
			3.00000000e+00 7.07671505e-02 \\ 
		};
		\addplot [line width = 1,color=red]
		table[row sep=crcr]{%
			6.99699868e-02 6.42836384e-01 \\ 
			5.00000000e-01 4.61167914e-02 \\ 
			5.50000000e-01 3.45505154e-02 \\ 
			6.00000000e-01 2.61517534e-02 \\ 
			6.50000000e-01 2.00377311e-02 \\ 
			7.00000000e-01 1.54924343e-02 \\ 
			7.50000000e-01 1.21821828e-02 \\ 
			8.00000000e-01 9.75199363e-03 \\ 
			8.50000000e-01 7.95414613e-03 \\ 
			9.00000000e-01 6.67301815e-03 \\ 
			9.50000000e-01 5.67417428e-03 \\ 
			1.00000000e+00 4.96608044e-03 \\ 
			1.05000000e+00 4.48463365e-03 \\ 
			1.10000000e+00 4.18433252e-03 \\ 
			1.15000000e+00 3.85428612e-03 \\ 
			1.20000000e+00 3.85316851e-03 \\ 
			1.25000000e+00 3.77112935e-03 \\ 
			1.30000000e+00 3.99487779e-03 \\ 
			1.35000000e+00 4.03524767e-03 \\ 
			1.40000000e+00 4.27891600e-03 \\ 
			1.45000000e+00 4.72094401e-03 \\ 
			1.50000000e+00 5.22989331e-03 \\ 
			1.60000000e+00 6.51788379e-03 \\ 
			1.70000000e+00 8.03448711e-03 \\ 
			1.80000000e+00 1.01596737e-02 \\ 
			1.90000000e+00 1.31456085e-02 \\ 
			2.00000000e+00 1.72486511e-02 \\ 
			2.10000000e+00 2.19925286e-02 \\ 
			2.20000000e+00 2.81171976e-02 \\ 
			2.30000000e+00 3.54703932e-02 \\ 
			2.40000000e+00 4.44617432e-02 \\ 
			2.50000000e+00 5.53345823e-02 \\ 
			2.60000000e+00 6.71495724e-02 \\ 
			2.70000000e+00 8.03515829e-02 \\ 
			2.80000000e+00 9.71776103e-02 \\ 
			2.90000000e+00 1.16327037e-01 \\ 
			3.00000000e+00 1.36509163e-01 \\ 
		};
		
		\addplot [line width = 1,color=blue, dashed, forget plot]
		table[row sep=crcr]{%
			6.49225917e-02 5.14561464e-01 \\ 
			5.00000000e-01 8.08235167e-03 \\ 
			5.50000000e-01 5.39086851e-03 \\ 
			6.00000000e-01 3.63340780e-03 \\ 
			6.50000000e-01 2.45953182e-03 \\ 
			7.00000000e-01 1.74788884e-03 \\ 
			7.50000000e-01 1.28492584e-03 \\ 
			8.00000000e-01 9.70600116e-04 \\ 
			8.50000000e-01 7.45662326e-04 \\ 
			9.00000000e-01 6.21706423e-04 \\ 
			9.50000000e-01 5.19416937e-04 \\ 
			1.00000000e+00 4.53490708e-04 \\ 
			1.05000000e+00 4.19648670e-04 \\ 
			1.10000000e+00 3.92199503e-04 \\ 
			1.15000000e+00 3.96365296e-04 \\ 
			1.20000000e+00 4.13288152e-04 \\ 
			1.25000000e+00 4.34569360e-04 \\ 
			1.30000000e+00 4.61889897e-04 \\ 
			1.35000000e+00 5.09632392e-04 \\ 
			1.40000000e+00 5.83775730e-04 \\ 
			1.45000000e+00 6.60751437e-04 \\ 
			1.50000000e+00 7.71737349e-04 \\ 
			1.60000000e+00 1.08521754e-03 \\ 
			1.70000000e+00 1.57532827e-03 \\ 
			1.80000000e+00 2.27794545e-03 \\ 
			1.90000000e+00 3.26768593e-03 \\ 
			2.00000000e+00 4.79180500e-03 \\ 
			2.10000000e+00 6.94750263e-03 \\ 
			2.20000000e+00 9.45374008e-03 \\ 
			2.30000000e+00 1.30248261e-02 \\ 
			2.40000000e+00 1.76066671e-02 \\ 
			2.50000000e+00 2.37672202e-02 \\ 
			2.60000000e+00 3.18970703e-02 \\ 
			2.70000000e+00 4.06272679e-02 \\ 
			2.80000000e+00 4.96256491e-02 \\ 
			2.90000000e+00 6.20522114e-02 \\ 
			3.00000000e+00 7.59335683e-02 \\ 
		};
		\addplot [line width = 1,color=blue,dashed]
		table[row sep=crcr]{%
			 6.49225917e-02 6.62916949e-01 \\ 
			5.00000000e-01 4.56555922e-02 \\ 
			5.50000000e-01 3.44663079e-02 \\ 
			6.00000000e-01 2.61331076e-02 \\ 
			6.50000000e-01 1.99043643e-02 \\ 
			7.00000000e-01 1.54740282e-02 \\ 
			7.50000000e-01 1.22226545e-02 \\ 
			8.00000000e-01 9.81246897e-03 \\ 
			8.50000000e-01 7.93632558e-03 \\ 
			9.00000000e-01 6.70518302e-03 \\ 
			9.50000000e-01 5.72061362e-03 \\ 
			1.00000000e+00 4.97713165e-03 \\ 
			1.05000000e+00 4.51219149e-03 \\ 
			1.10000000e+00 4.13547498e-03 \\ 
			1.15000000e+00 3.97433796e-03 \\ 
			1.20000000e+00 3.93533160e-03 \\ 
			1.25000000e+00 3.91571482e-03 \\ 
			1.30000000e+00 3.89768355e-03 \\ 
			1.35000000e+00 4.06862909e-03 \\ 
			1.40000000e+00 4.30982761e-03 \\ 
			1.45000000e+00 4.61354283e-03 \\ 
			1.50000000e+00 4.99508404e-03 \\ 
			1.60000000e+00 6.18113284e-03 \\ 
			1.70000000e+00 7.72061180e-03 \\ 
			1.80000000e+00 9.95447418e-03 \\ 
			1.90000000e+00 1.28007089e-02 \\ 
			2.00000000e+00 1.67171011e-02 \\ 
			2.10000000e+00 2.16911878e-02 \\ 
			2.20000000e+00 2.72498711e-02 \\ 
			2.30000000e+00 3.44510035e-02 \\ 
			2.40000000e+00 4.33284972e-02 \\ 
			2.50000000e+00 5.39294560e-02 \\ 
			2.60000000e+00 6.78103125e-02 \\ 
			2.70000000e+00 8.12509587e-02 \\ 
			2.80000000e+00 9.55103396e-02 \\ 
			2.90000000e+00 1.13101896e-01 \\ 
			3.00000000e+00 1.32631676e-01 \\ 
		};
		
		\addplot [color=black, dashdotted, line width = 1, forget plot]
		table[row sep=crcr]{%
			0	0.001\\
			3	0.001\\
		};
		\addplot [color=black, dashdotted, line width = 1]
		table[row sep=crcr]{%
			0	0.01\\
			3	0.01\\
		};
		
		\node at (axis cs:0.45, 5e-4)  (class1)    {Class $1$};
		\node at (axis cs:1.12, .018)  (class2)    {Class $2$};
		
		\node at (axis cs:3.15, .001)  ()    {$\epsilon_1$};
		\node at (axis cs:3.15, .01)  ()    {$\epsilon_2$};
		
	\end{axis}
\end{tikzpicture}%
 \label{subfig:AVPopt4}} \hspace{-.1cm}
		\vspace{-.4cm}
	\caption{The \gls{AVP} vs. $\Phi$ for the optimized distributions for the scenario in Example~\ref{example}. Red solid lines represent the optimized distributions with approximate PLR shown in Table~\ref{tab:optimized_dist}. Blue dashed lines represent the optimized distributions with simulated PLR.}
	\label{fig:AVPopt_vs_power}
		\vspace{-.45cm}
\end{figure}

\vspace{-.1cm}
\section{Conclusion} \label{sec:conclusion}
\vspace{-.1cm}
We investigated the trade-off between the \gls{AVP} and power consumption in a status-update system with multiple classes of users operating according to the \gls{IRSA} protocol. Specifically, we illustrate the benefits of jointly optimizing the update probability and the degree distributions of each class to minimize the average number of transmitted packets per slot. To perform this optimization efficiently, we proposed an easy-to-compute \gls{PLR} approximation, which yields an accurate approximation of the \gls{AVP}. Our simulation results suggest that irregular distributions are needed, and degrees up to $3$ are sufficient for the considered setting.

\vspace{-.1cm}

\bibliographystyle{IEEEtran}
\bibliography{IEEEabrv,./biblio}
\end{document}